\documentclass[lettersize,journal]{IEEEtran}
\IEEEoverridecommandlockouts
\usepackage{cite}
\usepackage[caption=false,font=normalsize,labelfont=sf,textfont=sf]{subfig}
\usepackage{amsthm,amsmath,amssymb,amsfonts}
\usepackage{algorithmic}
\usepackage{graphicx}
\usepackage{textcomp}
\usepackage{xcolor}
\def\BibTeX{{\rm B\kern-.05em{\sc i\kern-.025em b}\kern-.08em
		\kern-.1667em\lower.7ex\hbox{E}\kern-.125emX}}

\usepackage{url}
\usepackage[T1]{fontenc}

\usepackage{soulutf8}
\usepackage{tikz}
\usepackage{pgfplots}
\usepackage{bm}
\usepackage{cite}
\usepackage{acro} \acsetup{format/foreign = {\emph}}
\usepackage[capitalise]{cleveref}

\newtheorem{conjecture}{Conjecture}

\newtheorem{theorem}{Theorem}
\newtheorem{corollary}{Corollary}
\newtheorem{definition}{Definition}
\newtheorem{lemma}{Lemma}

\usepackage{soul}
\soulregister\ac7

\DeclareAcronym{PLS}{
	short = PLS,
	long  = physical layer security,
	tag = acrs,
}

\DeclareAcronym{FDM}{
	short = FDM,
	long  = frequency-division multiplexing,
	tag = acrs,
}
\DeclareAcronym{OFDM}{
	short = OFDM,
	long  = orthogonal frequency-division multiplexing,
	tag = acrs,
}
\DeclareAcronym{OFDMA}{
	short = OFDMA,
	long  = orthogonal frequency-division multiple access,
	tag = acrs,
}
\DeclareAcronym{LTE}{
	short = LTE,
	long  = long term evolution,
	tag = acrs,
}
\DeclareAcronym{4G}{
	short = 4G,
	long  = fourth generation of mobile network,
	tag = acrs,
}
\DeclareAcronym{5G}{
	short = 5G,
	long  = fifth-generation,
	tag = acrs,
}
\DeclareAcronym{SDMA}{
	short = SDMA,
	long  = space division multiple access,
	tag = acrs,
}
\DeclareAcronym{SDMA-OFDM}{
	short = SDMA-OFDM,
	long  = space division multiple access orthogonal frequency-division multiplexing,
	tag = acrs,
}
\DeclareAcronym{NOMA}{
	short = NOMA,
	long  = non-orthogonal multiple access,
	tag = acrs,
}
\DeclareAcronym{MIMO}{
	short = MIMO,
	long  = multiple-input-multiple-output,
	tag = acrs,
}
\DeclareAcronym{MIMO-NOMA}{
	short = MIMO-NOMA,
	long  = multiple-input-multiple-output non-orthogonal multiple access,
	tag = acrs,
}
\DeclareAcronym{CNN}{
	short = CNN,
	long  = convolutional neural network,
	tag = acrs,
}
\DeclareAcronym{LOS}{
	short = LOS,
	long  = line-of-sight,
	tag = acrs,
}
\DeclareAcronym{NLOS}{
	short = NLOS,
	long  = non-line-of-sight,
	tag = acrs,
}
\DeclareAcronym{WINNER}{
	short = WINNER,
	long  = Wireless World Initiative for New Radio,
	tag = acrs,
}
\DeclareAcronym{QuaDRiGa}{
	short = QuaDRiGa,
	long  = quasi deterministic radio channel generator,
	tag = acrs,
}
\DeclareAcronym{MT}{
	short = MT,
	long  = mobile terminal,
	tag = acrs,
}
\DeclareAcronym{BS}{
	short = BS,
	long  = base station,
	tag = acrs,
}
\DeclareAcronym{NR}{
	short = NR,
	long  = new radio,
	tag = acrs,
}
\DeclareAcronym{AWGN}{
	short = AWGN,
	long  = additive white Gaussian noise,
	tag = acrs,
}
\DeclareAcronym{SNR}{
	short = SNR,
	long  = signal-to-noise ratio,
	tag = acrs,
}
\DeclareAcronym{ADAM}{
	short = ADAM,
	long  = adaptive moment estimation,
	tag = acrs,
}
\DeclareAcronym{MSE}{
	short = MSE,
	long  = mean squared error,
	tag = acrs,
}
\DeclareAcronym{MAC}{
	short = MAC,
	long  = media access control,
	tag = acrs,
}
\DeclareAcronym{BER}{
	short = BER,
	long  = bit error rate,
	tag = acrs,
}
\DeclareAcronym{WMSN}{
	short = WMSN,
	long  = wireless multimedia sensor network,
	tag = acrs,
}
\DeclareAcronym{WSN}{
	short = WSN,
	long  = wireless sensor network,
	tag = acrs,
}
\DeclareAcronym{BTC}{
	short = BTC,
	long  = brain-type communications,
	tag = acrs,
}
\DeclareAcronym{IoT}{
	short = IoT,
	long  = Internet of Things,
	tag = acrs,
}
\DeclareAcronym{PW-MAC}{
	short = PW-MAC,
	long  = predictive-wakeup MAC,
	tag = acrs,
}
\DeclareAcronym{TDMA}{
	short = TDMA,
	long  = time-division multiple access,
	tag = acrs,
}
\DeclareAcronym{FDMA}{
	short = FDMA,
	long  = frequency-division multiple access,
	tag = acrs,
}
\DeclareAcronym{CDMA}{
	short = CDMA,
	long  = Ccde-division multiple access,
	tag = acrs,
}
\DeclareAcronym{BPSK}{
	short = BPSK,
	long  = binary phase-shift keying,
	tag = acrs,
}
\DeclareAcronym{QPSK}{
	short = QPSK,
	long  = quadrature phase-shift keying,
	tag = acrs,
}
\DeclareAcronym{MPSK}{
	short = MPSK,
	long  = $M$-ary phase-shift keying,
	tag = acrs,
}
\DeclareAcronym{QAM}{
	short = QAM,
	long  = quadrature amplitude modulation,
	tag = acrs,
}
\DeclareAcronym{RV}{
	short = RV,
	long  = random variable,
	tag = acrs,
}
\DeclareAcronym{pdf}{
	short = PDF,
	long  = probability density function,
	tag = acrs,
}
\DeclareAcronym{pmf}{
	short = PMF,
	long  = probability mass function,
	tag = acrs,
}
\DeclareAcronym{MAP}{
	short = MAP,
	long  = maximum a posteriori,
	tag = acrs,
}
\DeclareAcronym{MTE}{
	short = MTE,
	long  = meantime between events,
	tag = acrs,
}
\DeclareAcronym{SFC}{
	short = SFC,
	long  = semantic-functional communication,
	tag = acrs,
}
\DeclareAcronym{OOK}{
	short = OOK,
	long  = on-off keying,
	tag = acrs,
}
\DeclareAcronym{DM}{
	short = DM,
	long  = directional modulation,
	tag = acrs,
}
\DeclareAcronym{OSI}{
	short = OSI,
	long  = Open System Interconnect,
	tag = acrs,
}
\DeclareAcronym{DDM}{
	short = DDM,
	long  = dynamic directional modulation,
	tag = acrs,
}
\DeclareAcronym{LPDDM}{
	short = LPDDM,
	long  = low-power dynamic directional modulation,
	tag = acrs,
}
\DeclareAcronym{EVM}{
	short = EVM,
	long  = error vector magnitude,
	tag = acrs,
}
\DeclareAcronym{ASC}{
	short = ASC,
	long  = average secrecy capacity,
	tag = acrs,
}
\DeclareAcronym{SOP}{
	short = SOP,
	long  = secrecy outage probability,
	tag = acrs,
}
\DeclareAcronym{IQ}{
	short = IQ,
	long  = in-phase and quadrature,
	tag = acrs,
}
\DeclareAcronym{mmWave}{
	short = mmWave,
	long  = millimeter-wave,
	tag = acrs,
}
\DeclareAcronym{UAV}{
	short = UAV,
	long  = unmanned aerial vehicles,
	tag = acrs,
}
\DeclareAcronym{PA}{
	short = PA,
	long  = phased arrays,
	tag = acrs,
}
\DeclareAcronym{FDA}{
	short = FDA,
	long  = frequency diverse array,
	tag = acrs,
}
\DeclareAcronym{ZF}{
	short = ZF,
	long  = zero-forcing,
	tag = acrs,
}
\DeclareAcronym{RF}{
	short = RF,
	long  = radio frequency,
	tag = acrs,
}
\DeclareAcronym{SR}{
	short = SR,
	long  = secrecy rate,
	tag = acrs,
}
\graphicspath{{Figures/}}
\pgfplotsset{compat=1.17}

\usepackage{gensymb}

\begin{document}

\title{Low-Complexity Dynamic Directional Modulation: Vulnerability and Information Leakage}

\author{Pedro E. Gória Silva,
Adam Narbudowicz, \textit{Senior Member}, \textit{IEEE},
Nicola Marchetti, \textit{Senior Member}, \textit{IEEE}, \\
Pedro H. J. Nardelli,  \textit{Senior Member}, \textit{IEEE}, 
Rausley A. A. de Souza, \textit{Senior Member}, \textit{IEEE}, \\
and
Jules M. Moualeu, \textit{Senior Member}, \textit{IEEE}

\thanks{P. E. G. Silva and P. H. J. Nardelli  are with Lappeenranta--Lahti University of Technology, Finland (email: pedro.goria.silva@lut.fi, pedro.nardelli@lut.fi). 
P. E. G. Silva is also with INATEL, Brazil.
P. H. J. Nardelli  is also with University of Oulu, Finland. \\
A. Narbudowicz and N. Marchetti are with Trinity College Dublin, Ireland (email:{narbudoa; nicola.marchetti}@tcd.ie).
A. Narbudowicz is also working part-time with Wroclaw University of Science and Technology, Wroclaw, Poland.\\
R. A. A. de Souza is with National Institute of Telecommunications (Inatel), Santa Rita do Sapucaí 37540-000, Brazil (e-mail: rausley@inatel.br).\\
J. M. Moualeu is with the University of the Witwatersrand, Johannesburg, South Africa (e-mail: jules.moualeu@wits.ac.za). J. M. Moualeu is currently with Lappeenranta--Lahti University of Technology on a research visit.\\
This paper is partly supported by Academy of Finland via: (a) FIREMAN consortium n.326270 as part of CHIST-ERA grant CHIST-ERA-17-BDSI-003,  (b) EnergyNet Fellowship n.321265/n.328869/n.352654, (c) X-SDEN project n.349965, and (d) Science Foundation Ireland under grant 13/RC/2077\_P2 (CONNECT); by  CNPq (311470/2021-1); by São Paulo Research Foundation (FAPESP) (Grant No. 2021/06946-0); by RNP, with resources from MCTIC, Grant No. 01245.010604/2020-14, under the Brazil 6G project of the Radiocommunication Reference Center (\textit{Centro de Referência em Radiocomunicações} - CRR) of the National Institute of Telecommunications (\textit{Instituto Nacional de Telecomunicações} - Inatel), Brazil;
by the National Research Foundation (NRF) of South Africa under the BRICS Multilateral Research and Development Project (Grant No. 116018); and by  Business Finland under the project REEVA (n.10278/31/2022).
}}

\maketitle
\begin{abstract}
In this paper, the privacy of wireless transmissions is improved through the use of an efficient technique termed \ac{DDM}, and is subsequently assessed in terms of the measure of information leakage. 
Recently, a variation of \ac{DDM} termed \ac{LPDDM} has attracted significant attention as a prominent secure transmission method due to its ability to further improve the privacy of wireless communications. 
Roughly speaking, this modulation operates by randomly selecting the transmitting antenna from an antenna array whose radiation pattern is well known. 
Thereafter, the modulator adjusts the constellation phase so as to ensure that only the legitimate receiver recovers the information. 
To begin with, we highlight some privacy boundaries inherent to the underlying system. 
In addition, we propose features that the antenna array must meet in order to increase the privacy of a wireless communication system. 
Last, we adopt a uniform circular monopole antenna array with equiprobable transmitting antennas in order to assess the impact of \ac{DDM} on the information leakage. 
It is shown that the \acl{BER}, while being a useful metric in the evaluation of wireless communication systems, does not provide the full information about the vulnerability of the underlying system.

\end{abstract}

\acresetall

\section{Introduction} \label{sec: Intro}
\IEEEPARstart{I}{n} recent years, \ac{PLS} has attracted much attention from the research community as a promising technology for securing communications over wireless channels (see \cite{Wu8335290,Jameel8437135,Moualeu8734105} and the references therein). 
In contrast to the traditional approach, which addresses information security through mathematically derived encryption techniques in the upper layers of the wireless network protocol stacks, \ac{PLS} or information-theoretic security is implemented at the lowest layer of the \ac{OSI} stack, i.e., the physical layer. 
It exploits the unique physical features of wireless propagation channels to mitigate the amount of information obtained by unintended receivers through eavesdropping or malicious attacks \cite{Zou7467419,Liu7539590,Hamamreh8509094,Narbudowicz9674846}.
Owing to its promising benefits, \ac{PLS} has been widely investigated in the existing literature as a potential candidate to safeguard \ac{5G} and beyond wireless communications, and is intended to complement the encryption-based method. 
However, there are drawbacks associated with the \ac{PLS} implementation, such as the significant memory consumption, which may impose strict requirements on wireless devices in the context of \ac{5G} technology \cite{Cheng9351767}. 

A recent \ac{PLS} technique referred to as \ac{DM} (see \cite{Daly5159486,Ding6746064,Narbudowicz7909004,Ding_Fusco_2015} and the references therein) has emerged as an efficient and secure transmission approach suitable for wireless communications, including \ac{mmWave}, \ac{UAV}, satellite communication, and smart transportation \cite{Wang8334230,NusenuShaddrackYaw2019DoFM}.
\ac{DM} has the potential to steer the intelligent base-band information in the desired direction while transmitting distorted signals in other directions. 
Specifically, the original information-bearing signal is only transmitted in a narrow directive beam toward the intended user. 
The genesis of \ac{DM} stems from \cite{Daly5159486}, which uses \ac{PA} to improve security as long as the eavesdropper is not in the same direction as the legitimate receiver.

Recent studies have revealed a high degree of complexity as a result of the combination of the antenna array and complex signal processing techniques. 
For instance, the authors in \cite{Xie8103767} propose a \ac{ZF} technique to reduce the computational complexity of \ac{DM}; although the proposed solution addresses the complexity issue, it comes at the expense of a large antenna array. 
The works \cite{Narbudowicz7909004, Parron9140321} propose means to miniaturize the antenna array but at the cost of hardware complexity since the system requires individual amplitude and phase control over each port and additional \ac{RF} chains.
The authors in \cite{Huang9345966} synthesize multi-carrier \acl{DM} symbols for \ac{PLS} through a meticulous design of a time-switching sequence and time-modulated \ac{PA}. 
However, this approach is impractical for \ac{IoT} devices or other low-cost devices, because of the high complexity deriving from the synchronization with high-degree accuracy in the switching mechanism \cite{Bogdan8657780}. 
Similar to \cite{Huang9345966}, the work \cite{Narbudowicz9674846} proposes a practical and simple scheme for \ac{IoT} devices that provides energy efficiency and cost-effectiveness through a single \ac{RF} chain and an array of closely spaced antennas.
It requires a switchable antenna array and a random number generator at the transmitter along with a simple receiver. Moreover, it does not introduce additional artificial noise in the direction of the legitimate receiver. 

Despite the growing interest in \ac{DM}-based \ac{PLS} transmissions, only a limited number of metrics have been evaluated in the existing literature. 
In \cite{Ding6746064}, the authors assess various metrics such as the \ac{EVM}, \ac{BER}, and \ac{SR}. 
The \ac{EVM} is commonly adopted to quantify the system performance without executing the demodulation process and to distinguish physical sources from distortion. 
To summarize, \ac{EVM} calculates the normalized mean of the square of the difference between the measured streams and the reference symbols in the \ac{IQ} space. 
Albeit the \ac{BER}, \ac{EVM}, and \ac{SR} are equivalent for dynamic \ac{DM} systems under the assumption of a zero-mean Gaussian distributed orthogonal interference, the authors in \cite{Ding6746064} point out a discrepancy between such metrics for static \ac{DM} systems. 
Recently, several studies have used the \ac{BER} or a variation of the \ac{SR} as a metric to evaluate the system performance \cite{Cheng9351767,Narbudowicz7909004,Narbudowicz9674846,Qu9664476}. 
In  \cite{Cheng9351767}, the system performance is evaluated in terms of the \ac{ASC}, \ac{SOP}, and \ac{BER}. 
It is worthwhile recalling the dependence of both the \ac{ASC} and \ac{SOP} on the \ac{SR}, which is intrinsically linked to the mutual information. 
However, the mutual information is not an adequate security metric\footnote{This will be subsequently discussed in the manuscript.}. 
A closely related work to our proposed study is \cite{Narbudowicz9674846} wherein a uniform circular array of mono-pole antennas with equiprobable transmitting antennas is proposed. 
In that work, Narbudowicz \textit{et al.} show that the privacy of a wireless communication can be achieved by randomly switching the transmitting antenna. 
Furthermore, the authors rely on the \ac{BER} as a performance metric to evaluate the system's privacy. 
%

Motivated by the preceding discussion, the proposed work aims at adopting other performance metrics as a means to evaluate the security of a wireless communication system in terms of vulnerability. 
Our stimulus is prompted in a way by the inefficacy of the \ac{BER} in gauging the system vulnerability. 
To the best of the authors' knowledge, our work is the first that adopts a theoretical approach to information and security on the new \ac{DM} scheme. 
The contributions of this work are as follows:
\begin{enumerate}
    \item[$\bullet$] We highlight some privacy boundaries inherent to the underlying system and provide useful insights on the way an eavesdropper can intercept confidential information intended for a legitimate receiver.
    \item[$\bullet$] We propose the design of an antenna array that minimizes the system vulnerability for an acceptable level of privacy. Moreover, we derive the fundamental limits of the system vulnerability.
    \item[$\bullet$] We adopt a uniform circular mono-pole antenna array with equiprobable transmitting antennas and then evaluate their impact on the measure of information leakage.
    \item[$\bullet$] We understand how the transmission of constellations can ensure privacy in the context of data confidentiality.
    \item[$\bullet$] We show that the system vulnerability can be significantly reduced for a large number of transmitting antennas.
\end{enumerate}


The rest of the paper is organized as follows. The description of the system model is provided in \cref{sec: Sys}, while the proposed vulnerability approach is elaborated in \cref{sec: VA}. \cref{sec: Design} presents the theoretical limits of the proposed system and discusses the antenna-array design that minimizes the information leakage. A case study is investigated and numerical results are illustrated in \cref{sec: Case Study}. Finally, \cref{sec: Conclusion} provides some concluding remarks.

%
%
%

\section{System Model} \label{sec: Sys}

\begin{figure*}[ht!]
    \centering
    \includegraphics[width=1\linewidth]{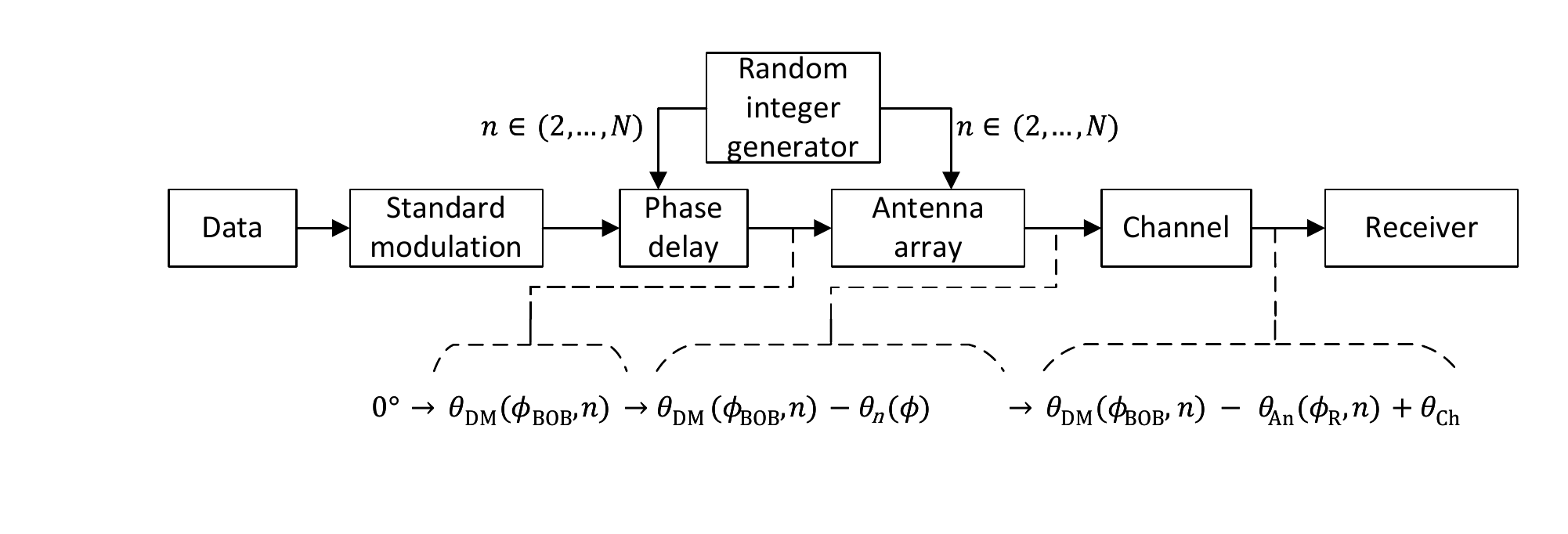}
    \caption{Block diagram of the energy efficient DM system.}  
    \label{fig: system}
\end{figure*}
%

The energy-efficient \ac{DM} system proposed in \cite{Narbudowicz9674846} consists of a standard modulator, a phase delayer, a circular antenna array, and a random integer generator.
The block diagram of the system and the phase displacement are presented in \cref{fig: system}.
The standard modulator block performs \ac{IQ} modulation of the data. 
It is worth mentioning that phase-based modulation (for instance, \acl{MPSK} or \acl{QAM}) represents the worst-case scenario for eavesdroppers to perform a successful attack; in brief, the security of the \ac{DM} scheme is based on phase distortion, and thus modulations whose information is not somehow contained in phase are more susceptible to successfully attack.
It is assumed that the phase displacement at the modulator output is null for the sake of simplicity. 
Let $\theta_{n}(\phi)$ be the phase shift generated by the radiation pattern (including the antenna's displacement) of the $n$-th antenna in the direction $\phi$.
The phase delay rotates the symbol by $\theta_{\text{DM}}(\phi_{\text{Bob}},n)$ radians, with $\phi_{\text{Bob}}$ and $n \in \{1, 2, \dots, N\}$ being the direction of the legitimate (intended) receiver and the transmitting antenna, respectively. For the proposed \ac{DM} scheme, this phase shift added in the phase delay block is given by
\begin{equation}\label{eq:DMphase}
\theta_{\text{DM}}(\phi_{\text{Bob}},n) = - \theta_{n}(\phi_{\text{Bob}}).
\end{equation}

The phase delay block works as a pre-distortion phase of the transmitted signal, which allows its correct reconstruction only along the intended direction $ \phi_{\text{Bob}} $.
Therefore, the total phase shift at the receiver is given by

\begin{equation}\label{eq:total_phase}
\theta_\text{Rx}(\phi, n) = \theta_{\text{DM}}(\phi_{\text{Bob}},n) + \theta_{n}(\phi) + \theta_{\text{Ch},n},
\end{equation}
where $\theta_{\text{Ch},n}$ is defined as the phase shift of the channel between the $n$th antenna and receiver, as indicated in \cref{fig: system}. 
Hereafter, we will omit the $\theta_{\text{Ch,n}}$ dependence on the transmitting antenna $n$ in order to simplify the notation. 
Note that $\theta_{\text{Ch},n} \forall n \in {1,2,\dots,N}$ are statistically identical.
After substituting (\ref{eq:DMphase}) into (\ref{eq:total_phase}), it can be seen that the phase shift from the phase delay block and the antenna array cancel out each other in the direction $\phi = \phi_{\text{Bob}}$, leaving the phase shift at the legitimate receiver to be caused exclusively by the channel delay. 
On the other hand, there is an additional phase component that changes as a function of the direction $\phi$, leading to the distortion of the transmitted signal for $\phi \neq \phi_{\text{Bob}}$ --- with the severity of the distortion being dependent on the antenna-array design.

However, if the eavesdropper obtains information about the radiation pattern used, it can somehow compensate for the additional phase shift.
The system aims to preclude this by randomly switching between different antennas in the array, i.e., it is the primary function of the random integer generator. 
In this way, the eavesdropper cannot squarely distinguish the phase artificially generated by the \ac{DM} system from the phase modulation. 
Therefore, the system may provide a certain level of security for the secret information while transmitting via only one antenna at each given time; the other antennas of the array remain disconnected. 
Thereby, we can use a single RF chain leading to a significant hardware simplification compared to other \ac{DM} schemes \cite{Narbudowicz9674846}.
Before the data transmission phase, most wireless systems require a handshake. 
Here, the transmitter and the legitimate receiver perform a handshake through a reference antenna $n_{\text{R}}$, which is not used later for data transmission.
Hereafter, we will assume that the antenna array is composed of $N+1$ antennas in such a way that we have transmitting antenna $n \in \{1,2,\dots,N\}$; in the sequel,
when referring to the antenna array, we allude to the $N$ antennas usable in the transmission phase after the handshake.

\section{Vulnerability approach}
\label{sec: VA}
%
In this section, we briefly recall the concept of vulnerability as a measure of the information leakage proposed by Smith in \cite{Smith10.1007/978-3-642-00596-1_21} (closely related to Bayes risk).
The measure of information leakage or vulnerability --- which is, in a way, the opposite of uncertainty --- uses Rényi's \textit{min-entropy} rather than Shannon entropy.
%

In secure information flow analysis, the question that arises is whether a broadcast transmission could leak information about confidential data to unintended receivers. 
If there is indeed an information leakage, then the main concern is how to measure it.
Intuitively, one can expect that ``\textit{initial uncertainty = information leaked + remaining uncertainty}''. 
This claim plainly suggests Shannon's traditional mutual information as the measure. 
However, \cite{Smith10.1007/978-3-642-00596-1_21} suggests that Shannon's mutual information may be inadequate to measure information leakage.
A more efficient way to measure information leakage would be based on the notion of \textit{vulnerability} and \textit{min-entropy}~\cite{Smith10.1007/978-3-642-00596-1_21}.
We briefly revise these concepts here.

Let the triple $(\mathcal{S}, \mathcal{O}, \mathbf{C})$ represents a discrete channel, in which $\mathcal{S}$ is a finite set of secret input values, $\mathcal{O}$ is a finite set of observable output values, and $\mathbf{C}$ is a $|\mathcal{S}| \times |\mathcal{O}|$ channel matrix, such that the element $c_{i,j}$ in the $i$th row and the $j$th column represents the conditional probability of obtaining the $i$th element of $\mathcal{O}$ in the output given that the input is the $j$th element of $\mathcal{S}$. 
Assuming any \textit{a priori} \ac{pmf} $p_S(s)$ on $\mathcal{S}$, we have the \ac{RV} $S$ that represents secret input values.
Furthermore, taking $p_S(s)$ and $\mathbf{C}$, we can define the output \ac{RV} $O$ with a \ac{pmf} given by $p_O(o)=\sum_{s} p_O(o|s) p_S(s)$.

\begin{definition}[Vulnerability]\label{def:Vul}
    Suppose a potential eavesdropper $\mathcal{A}$ wishes to guess the value of $S$ \textit{in one try} and, in the worst-case assumption, admit that $\mathcal{A}$ knows $p_S(s)$ and $\mathbf{C}$; therefore, the \textit{a priori} vulnerability is mathematically formulated as
    \begin{equation}\label{eq: a pri vul}
        V(S) = \max_{s \in \mathcal{S}} p_S(s),
    \end{equation}
    and the \textit{a posteriori} vulnerability is given by
    \begin{IEEEeqnarray}{lCl}\label{eq: a pos vul}
        V(S|O) &=& \sum_{o \in \mathcal{O}} p_O(o) \max_{s \in \mathcal{S}}p_S(s|o) \nonumber \\
        &=& \sum_{o \in \mathcal{O}} \max_{s \in \mathcal{S}}(p_O(o|s)p_S(s)),
    \end{IEEEeqnarray}
    with $0<V(S) \leq V(S|O)\leq1$.
\end{definition}

We obtain the uncertainty measure by taking the negative logarithm of vulnerability $V(\cdot)$, i.e., the Rényi’s \textit{min-entropy}~\cite{Smith10.1007/978-3-642-00596-1_21,Renyi1960}.
Thus, we have the following definitions.
\begin{definition}[Initial Uncertainty]
    The initial uncertainty is given by
    \begin{equation}\label{eq: IU}
        H_{\infty}(S) = -\log_2(V(S)).
    \end{equation}
\end{definition}
\begin{definition}[Remaining Uncertainty]
    The remaining uncertainty is given by
    \begin{equation}\label{eq: RU}
        H_{\infty}(S|O) = -\log_2(V(S|O)).
    \end{equation}
\end{definition}
\begin{definition}[Information Leakage]\label{def: IL}
    The information leakage is given by the difference between initial uncertainty and remaining uncertainty. In mathematical terms, we have
    \begin{IEEEeqnarray}{lCl}\label{eq: IL}
        I_{\infty}(S;O) = H_{\infty}(S)-H_{\infty}(S|O)=\log_2 \left( \frac{V(S|O)}{V(S)} \right).
    \end{IEEEeqnarray}
\end{definition}
\cref{def: IL} reveals how vulnerable $S$ is to one guess (in one try), given that eavesdropper $\mathcal{A}$ knows $O$.
On the other hand, one can assume multiple guesses for $\mathcal{A}$; hence, $I_{\infty}(S;O)$ might become inadequate~\cite{Boris5552653}.
Nevertheless, we can draw a tractable bound for at most $g$-guesses made by $\mathcal{A}$ as a factor of $g$.
Thus, we have $V_g(S) \leq g V(S)$ and $V_g(S|O) \leq g V(S|O)$, with $V_g(S)$ being the vulnerability for $g$-guesses~\cite{Boris5552653}.

An optimal strategy certainly is to guess the value of $S$ according to a decreasing order of the probabilities of $S$ --- \cref{def:Vul} assume exactly the same if only one try is allowed.
Let $p_1,p_2,\dots,p_{|\mathcal{S}|}$ be the probabilities of $S$ such that $p_1 \geq p_2 \geq \dots \geq p_{|\mathcal{S}|}$.
Thus, the expected number of guesses required to guess $S$ optimally, i.e., the \textit{guessing entropy} of $S$, is given by
\begin{IEEEeqnarray}{lCl}\label{eq: G(S)}
    G(S) = \sum_{l=1}^{|\mathcal{S}|}lp_l.
\end{IEEEeqnarray}
%
%
Massey's guessing entropy bound in \cite{Massey394764} establishes a lower bound on the expected number of guesses required to find $S$ given $O$ as a function of Shannon's conditional entropy $H(S|O)$. 
It claims that the guessing entropy of $S$ given $O$ meets 
\begin{IEEEeqnarray}{lCl}
    G(S|O) \geq 2^{H(S|O) - 2} +1. 
\end{IEEEeqnarray}
However, this expected number of guesses may be arbitrarily large even when $H(S|O)$ --- and, accordingly, Massey's lower bound --- is arbitrarily small. 
Moreover, \cite{Smith10.1007/978-3-642-00596-1_21} indicates through examples that $G(S)$ or $G(S|O)$ can be high even when the attacker can guess $S$ reasonably well in one try.
A similar conclusion can be drawn by assuming $p_1$ and $|\mathcal{S}|$ large enough in \eqref{eq: G(S)}.
%
In contrast, since the conditional min-entropy $H_{\infty}(S|O)$ satisfies
\begin{equation}
    V(S|O) = 2^{-H_{\infty}(S|O)},
\end{equation}
it provides immediate security guarantees about the vulnerability of $S$ given $O$.
\section{Design}
\label{sec: Design}
This section deals with the theoretical limits of the underlying system. 
In addition, we quantify the vulnerability of the proposed system and discuss optimal characteristics for the eavesdropper and array of transmitting antennas.
In what follows, \cref{subs: Easves} assesses the way an eavesdropper can best usurp information from the transmitted signal, while \cref{subs: ArryaAnt} discusses the design of an antenna array that minimizes the information leakage.
\subsection{Eavesdropper}\label{subs: Easves}
Suppose eavesdropper $\mathcal{A}$ is in the direction $\phi_{\text{E}}$ with respect to the antenna array.
It follows that the phase of the local oscillator of $\mathcal{A}$ after the handshake is given by
\begin{equation}
    \rho_{\text{LO}} = \theta_{\text{DM}}(\phi_{\text{Bob}},n_{\text{R}}) + \theta_{n_{\text{R}}}(\phi_{\text{E}}) + \mathbb{E}[\theta_{\text{Ch}}],
\end{equation}
with $\mathbb{E}[\cdot]$ denoting the expectation operator. 
If the channel is \ac{AWGN}, then we have $\mathbb{E}[\theta_{\text{Ch}}] = 0$.
The phase shift of symbols received by $\mathcal{A}$ after the handshake (confidential data) is given by
\begin{equation}
    \rho_{\text{S}} = \theta_{\text{DM}}(\phi_{\text{Bob}},n) + \theta_{n}(\phi_{\text{E}}) + \theta_{\text{Ch}},
\end{equation}
and the phase error in $\mathcal{A}$ can be written as 
\begin{IEEEeqnarray}{lCl}\label{eq: Phase Error A}
    \rho &=& \rho_{\text{S}}-\rho_{\text{LO}} \nonumber \\
    &=& \theta_{\text{DM}}(\phi_{\text{Bob}},n) + \theta_{n}(\phi_{\text{E}}) + \theta_{\text{Ch}} \nonumber \\
    &&- \theta_{\text{DM}}(\phi_{\text{Bob}},n_{\text{R}}) - \theta_{n_{\text{R}}}(\phi_{\text{E}}) - \mathbb{E}[\theta_{\text{Ch}}].
\end{IEEEeqnarray}
Finally, the complex received symbol is given by 
\begin{IEEEeqnarray}{lCl}\label{eq: RecSimb r}
    r = |m + w| \exp(i (\rho + \beta)),
\end{IEEEeqnarray}
where $m$ represents the complex transmitted symbol with phase $\beta$, $w$ is a complex Gaussian \ac{RV} with zero mean and variance $\sigma^2$, and $i=\sqrt{-1}$.

The eavesdropper $\mathcal{A}$ must maximize the information leakage $I_{\infty}(S;O)$, with $S$ being the secret symbol whose complex representation is given by $m$, and $O$ is the received symbol.
Note that the mapping of $S$ into $m$ depends on the adopted modulation.
The eavesdropper's \textit{demodulation function} (or demodulation map) can be defined as $d: \mathbb{C} \rightarrow  \mathcal{O}$, and, thus, the received symbol can be written as $o=d(r)$.
Thereby, the task of eavesdropper $\mathcal{A}$ is to choose a function $d(\cdot)$ that maximises the chance of guessing $S$; this can be mathematically formulated as
\begin{IEEEeqnarray}{lCl}
        \max_{d(r)} I_{\infty}(S;O)\\
        \text{s.t.} \; |\mathcal{S}| \leq |\mathcal{O}| \leq |\mathcal{S}|N, \nonumber
\end{IEEEeqnarray}
or, equivalently, 
\begin{IEEEeqnarray}{lCl}
        \max_{d(r)} V(S|O) \\
        \text{s.t.} \; |\mathcal{S}| \leq |\mathcal{O}| \leq |\mathcal{S}|N, \nonumber 
\end{IEEEeqnarray}
where the constraint $|\mathcal{S}| \leq |\mathcal{O}| \leq |\mathcal{S}|N$ is explained as follows.
In fact, the totality of symbols sent through the channel by the transmitter is less than or equal to $|\mathcal{S}|N$; 
therefore, a set of output symbols $\mathcal{O}^*$ with $|\mathcal{O}^*|>|\mathcal{S}|N$ has spare elements.
%
%
$|\mathcal{S}| \leq |\mathcal{O}|$ can be trivially understood, as the output set cannot be smaller than the input set without loss of information in this particular context.
As will become clear throughout the remainder of the paper, we commonly have $|\mathcal{O}| = |\mathcal{S}|N$.
%

\begin{lemma}
    \label{lem: Max V(S|O)}
    Assume that $p_{R,S}(r,s)$ is the joint \ac{pdf} of the complex received symbol $R$ and secret input values $S$.
    Let 
    $\mathcal{P}_{s}$ be a complex plane, such that $p_{R,S}(r,s)>p_{R,S}(r,s^*) \, \forall \, r \in \mathcal{P}_{s}$ and $s \neq s^*$. 
    Then, $V(S|O)$ is maximized if and only if $d(r)$ maps all $r \in \mathcal{P}_{s}$ to the same observable output value $o$.  
\end{lemma}
\begin{proof}
    Let the complex plane $\mathcal{R}$ be the sample space of $R$, and assume that $d(r)=o$ \; $\forall r \in \mathcal{R}_o$ such that $\mathcal{R}_o  \subset \mathcal{R}$ and $\mathcal{R}_o \cap \mathcal{R}_{i} = \emptyset \; \forall i\neq o$. 
    Given the characteristics of $d(r)$, we can then write
    \begin{IEEEeqnarray*}{lCl}
        p_O(o|s) = \int_{\mathcal{R}_o} p_R(r|s) \text{d}r,
    \end{IEEEeqnarray*}
    and
    \begin{IEEEeqnarray*}{lCl}
        p_O(o|s)p_S(s) = \int_{\mathcal{R}_o} p_{R,S}(r,s) \text{d}r.
    \end{IEEEeqnarray*}
    
    Now, let $s_o$ be the value of $s$ which maximizes a term of the following summation
    \begin{IEEEeqnarray*}{lCl}
        \sum_{o \in \mathcal{O}} \max_{s \in \mathcal{S}}(p_O(o|s)p_S(s)).
    \end{IEEEeqnarray*}
    Thus, we can write
    \begin{IEEEeqnarray*}{lCl}
        \max_{s \in \mathcal{S}} \left( \int_{\mathcal{R}_o} p_{R,S}(r,s) \text{d}r \right) = \\
        =\int_{\mathcal{R}_o \cap \mathcal{P}_{s_o}} p_{R,S}(r,s_o) \text{d}r + \int_{\mathcal{R}_o - \mathcal{P}_{s_o}} p_{R,S}(r,s_o) \text{d}r \\ \leq  \sum_{a \in \mathcal{S}} \int_{\mathcal{R}_o \cap \mathcal{P}_{a}} p_{R,S}(r,a) \text{d}r = \int_{\mathcal{R}_o} \max_{s \in \mathcal{S}} p_R(r|s)p_S(s) \text{d}r.
    \end{IEEEeqnarray*}
    And it follows that 
    \begin{IEEEeqnarray}{lCl}
        \sum_{o \in \mathcal{O}} \max_{s \in \mathcal{S}}(p_O(o|s)p_S(s)) \leq \int_\mathcal{R} \max_{s \in \mathcal{S}} p_R(r|s)p_S(s) \text{d}r
    \end{IEEEeqnarray}
    with the equality condition being met only if $\mathcal{R}_o = \mathcal{P}_{s_o} \; \forall o \in \mathcal{O}$.
\end{proof}

\begin{theorem}
    \label{th: optmal d(r)}
    %
    Let $\mathcal{P}_{s,n}$ be a complex plane, such that $p_{R,S,n}(r,s,n)>p_{R,S,n}(r,s^*,n^*) \, \forall \, r \in \mathcal{P}_{s,n}$ and $s \neq s^* \lor n \neq n^*$ with $p_{R,S,n}(r,s,n)$ being the joint \ac{pdf} of the complex received symbol $R$, secret input values $S$, and transmitting antenna $n$.
    The eavesdropper $\mathcal{A}$ can maximize the system vulnerability by designing the demodulation function $d_{\mathcal{A}}(r)$ in a such way that for all $r \in \mathcal{P}_{s,n}$, it is mapped to the same output $o$.
\end{theorem}
\begin{proof}
    One can easily extend the proof of \cref{lem: Max V(S|O)} by considering the secret input values with a combination of $s$ and $n$ as follows.
    Let the modulation function be $c:\mathcal{S} \times \{1,2,\dots,N\} \to \mathbb{C}$. 
    Then, we can define the transmitting symbol $M$ as a \ac{RV} given by $m=c_{\phi}(s,n)$ for direction $\phi$. 
    Furthermore, the sample space of $M$ contains at most $|\mathcal{S}|N$ elements, and the \ac{pmf} of $M$ is given by $p_M(m)=p_S(s)p_n(n)$.
    Thus, we can apply \cref{lem: Max V(S|O)} by considering $M$ as the secret input.
\end{proof}
%
%
One way to interpret \cref{th: optmal d(r)} is as follows: the vulnerability maximization for $\mathcal{A}$ refers to finding the optimal decision areas (or boundaries) for a finite output alphabet. 
For a practical approach, the use of \cref{th: optmal d(r)} solely requires knowledge of the joint \ac{pdf} of the complex received symbol $R$, the secret input values $S$, and the transmitting antenna $n$. 
%
%
Hereafter, we will assume that the eavesdropper has precise knowledge about $p_{R,S,n}(r,s,n)$ in order to establish the worst-case scenario regarding vulnerability.
Note that any mismatch between the prior $p_{R,S,n}(r,s,n)$ at $\mathcal{A}$ and the correct $p_{R,S,n}(r,s,n)$ can lead to information leakage less than or equal to that obtained with perfect prior knowledge of $p_{R,S,n}(r,s,n)$.

One can note that $d_{\mathcal{A}}(\cdot)$ is always a non-injective surjective function.
Furthermore, for a sufficiently high \ac{SNR}, we have a bijective function from $S$ and $n$ to $O$; therefore, the information leakage would be maximum given that $\mathcal{A}$ would be able to recover precisely $S$.
\subsection{Transmitting antenna array} \label{subs: ArryaAnt}
Here, we turn turn our attention to the transmitter, and specifically the transmitting antenna array.
In this subsection, we assess how the arrangement of transmitting antennas impacts the system vulnerability.
Assuming a worst case-scenario (from a security viewpoint) wherein the eavesdropper is in its optimal operation, we provide some guidelines on the design of an antenna array that minimizes the information leakage.

\begin{lemma}
    \label{lem: V(S|O)}
    For the eavesdropper $\mathcal{A}$ that uses the demodulation function $d_{\mathcal{A}}(r)$ according to \cref{th: optmal d(r)}, the vulnerability of the system is then given by
    \begin{IEEEeqnarray}{lCl}
        V(S|O) = N-\sum_{o \in \mathcal{O}}\epsilon_{o,s_o},
    \end{IEEEeqnarray}
    where $\epsilon_{o,s}= \Pr[d_{\mathcal{A}}(r) \neq o\cap S=s]$ is the error probability and $s_o=\arg \max_{s\in \mathcal{S}}p_O(o|s)p_S(s)$. 
\end{lemma}
\begin{proof}
    Let $\epsilon'_{o,s}= \Pr[d_{\mathcal{A}}(r) \neq o|S=s]$, then we have that
    \begin{IEEEeqnarray}{lCl}
        \label{eq: lem pO}
        p_O(o|s) = (1-\epsilon'_{o,s}).
    \end{IEEEeqnarray}
    Plugging \eqref{eq: lem pO} into \eqref{eq: a pos vul}, we obtain
    \begin{IEEEeqnarray}{lCl}
    \label{eq: lem V}
        V(S|O) &=& \sum_{o \in \mathcal{O}}(1-\epsilon'_{o,s_o}) p_S(s_o) \nonumber \\ 
        &=& N-\sum_{o \in \mathcal{O}}\epsilon_{o,s_o},
    \end{IEEEeqnarray}
    with $\epsilon_{o,s_o} = \epsilon'_{o,s_o} p_S(s_o)$ and $s_o=\arg \max_{s\in \mathcal{S}}p_O(o|s)p_S(s)$.
\end{proof}

\cref{lem: V(S|O)} reveals the dependency relationship between vulnerability and error probability $\epsilon_{o,s}$.
Let us now take a closer look at this error probability, $\epsilon_{o,s}$, which can be expressed as
\begin{IEEEeqnarray}{lCl}
\label{eq: ep ext}
    \hspace{-5mm}
    \epsilon_{o,s_o} = p_S(s_o) \left(1 - \sum_{n} \int_{\mathcal{P}_{s_o,n_o}} p_{R}(r|s_o,n) \text{d}r \, p_n(n) \right),
\end{IEEEeqnarray}
with $o=d_{\mathcal{A}}(c_{\mathcal{A}}(s_o,n_o))$ and $c_{\mathcal{A}}(\cdot,\cdot)$ being the modulation function in the direction of ${\mathcal{A}}$ defined as $c:\mathcal{S} \times \{1,2,\dots,N\} \to \mathbb{C}$ --- note that $m=c_{\mathcal{A}}(s,n)$.
Since $p_{R}(r|s_o,n)$ is 2$\sigma^2$-variance, complex Gaussian whose center is at $c_{\mathcal{A}}(s_o,n)$, one can notice that $\lim_{\sigma \to 0}V(S|O) = 1$ and $\lim_{\sigma \to 0} I_{\infty}(S;O) = H_{\infty}(S)$.
In other words, the information leakage of the system tends to its maximum value insofar as the \ac{SNR} increases.

\begin{theorem}
    \label{th: Inf lea}
    For equiprobable secret symbols and uniformly distributed phase $\theta_{An}$ (that is, the antenna array in conjunction with the Random Integer Generator produces a uniform phase shift), the information leakage to an eavesdropper in the direction $\phi \neq \phi_{\text{Bob}}$ is given by 
    \begin{IEEEeqnarray}{lCl}
        I_{\infty}(S;O) = \log_2 \left( \sum_{l=1}^{L} Q_1\left(\frac{m_l}{\sigma },\frac{r_{l-1}}{\sigma},\frac{r_l}{\sigma }\right) \right),
    \end{IEEEeqnarray}
    with the set of non-repeating elements in ascending order $\{m_1,m_2,\dots,m_L \}$ being the modules of complex symbols $m$, $\{r_1,r_2,\dots,r_{L-1} | r_1<r_2<\dots<r_{L-1}\}$ are the radii of the circular decision regions of the output symbols $O$, $r_0=0$, $r_L=\infty$, $Q_m(a,b)=\int_b^{\infty } \frac{1}{2} x \left(\exp  \left(-\left(a^2+x^2\right)\right)\right) \left(\frac{x}{a}\right)^{m-1} I_{m-1}(a x) \, \text{d}x$ is the Marcum Q-function~\cite{marcumq}, and $Q_m(a,b_0,b_1)=Q_m\left(a,b_0\right)-Q_m\left(a,b_1\right)$.
\end{theorem}
\begin{proof}
    Let $\theta_{An}[k]$ be the $k$th sample of $ \theta_{n}(\phi)$ over time; and assume that $\theta_{An}[k] \; \forall k \in \{1,2,\dots\}$ has a uniform distribution in the interval $[-\pi,\pi]$, i.e, $p_{\theta_{An}}(\theta_{An})=1/(2 \pi)$. 
    Since the phase error in the direction $\phi$ is given by $\rho = \theta_{\text{DM}}(\phi_{\text{Bob}},n) + \theta_{n}(\phi_{\text{E}}) + \theta_{\text{Ch}}-\rho_{\text{LO}}$, the \ac{pdf} of $\rho$ can be obtained as
    \begin{IEEEeqnarray}{lCl}
        p_{\rho}(\rho) &=& \sum_{\theta_{\text{DM}}}\int_{-\pi}^{\pi} p_{\rho}(\rho|\theta_{\text{DM}},\theta_{\text{Ch}} ) p_{\theta_{\text{DM}}}(\theta_{\text{DM}}) p_{\theta_{\text{Ch}}}(\theta_{\text{Ch}}) \text{d} \theta_{\text{Ch}} \nonumber \\
        &=& \sum_{\theta_{\text{DM}}}\int_{-\pi}^{\pi} \frac{1}{2 \pi} p_{\theta_{\text{DM}}}(\theta_{\text{DM}}) p_{\theta_{\text{Ch}}}(\theta_{\text{Ch}}) \text{d} \theta_{\text{Ch}} \nonumber \\
        &=& \frac{1}{2 \pi}.
    \end{IEEEeqnarray}
    %
    
    From \cref{lem: Max V(S|O)}, we have the optimal decision regions' boundaries as circles. 
    Let the radius of the decision regions be given by $\{r_1,r_2,\dots,r_{L-1} | r_1<r_2<\dots<r_{L-1}\}$, we have that 
    \begin{IEEEeqnarray*}{lCl}
        p_O(o=o_l|s) &=& \int_{r_{l-1}}^{r_l} \frac{|r|}{\sigma ^2} \exp \left(\frac{-m_l^2-|r|^2}{2 \sigma ^2}\right) I_0\left(\frac{m_l|r|}{\sigma ^2}\right) \text{d}r \\
        &=& Q_1\left(\frac{m}{\sigma },\frac{r_{l-1}}{\sigma},\frac{r_l}{\sigma }\right),
    \end{IEEEeqnarray*}
    with $r_0=0$, $r_L=\infty$, and $l=\{1,2,\dots,L-1\}$.
    
   \noindent Then, the vulnerability expression can be written as
    \begin{IEEEeqnarray*}{lCl}
        V(S|O) &=& \frac{1}{|\mathcal{S}|}\sum_{l=1}^{L} Q_1\left(\frac{m_l}{\sigma },\frac{r_{l-1}}{\sigma},\frac{r_l}{\sigma }\right).
    \end{IEEEeqnarray*}
    
   \noindent Finally, the information leakage is given by
    \begin{IEEEeqnarray*}{lCl}
        I_{\infty}(S;O) = \log_2 \left( \sum_{l=1}^{L} Q_1\left(\frac{m_l}{\sigma },\frac{r_{l-1}}{\sigma},\frac{r_l}{\sigma }\right) \right).
    \end{IEEEeqnarray*}
    
\end{proof}

\begin{corollary}
    \label{cl: leak to 0}
    For equiprobable secret symbols and uniformly distributed phase $\theta_{An}=\theta_{\text{DM}}(\phi_{\text{Bob}},n) + \theta_{n}(\phi)$, the information leakage in the direction $\phi \neq \phi_{\text{Bob}}$ is null in phase shift key modulation. 
\end{corollary}
\begin{proof}
    For the phase-shift keying modulation, we have $L=1$ in \cref{th: Inf lea}. Therefore, the mutual information simplifies to
    \begin{IEEEeqnarray*}{lCl}
        I_{\infty}(S;O) = \log_2 \left( Q_1\left(\frac{m_1}{\sigma },0,\infty \right) \right) = \log_2(1) =0.
    \end{IEEEeqnarray*}
\end{proof}
Among other ways, we can understand \cref{cl: leak to 0} as an optimal limit for the communication security of the proposed system.
\cref{cl: leak to 0} establishes the possibility of secure communication between the transmitter and the legitimate receiver.
In other words, from the point of view of information leakage, \cref{cl: leak to 0} points to the condition in which the proposed system becomes impossible to decipher.
Note that both an eavesdropper or legitimate receiver in the direction $\phi_{\text{Bob}}$ would be able to intercept the message even though $\theta_{An}$ follows a uniform distribution.

\section{Case Study} \label{sec: Case Study}
In this section, we consider a specific arrangement of antennas, in order to assess how the number of antennas and the shape of the constellation impact information leakage.

We adopt a circular antenna array, which offers a 360° field of view in the horizontal plane. 
If the phase centers of the antennas are evenly spread around a circle of diameter $D$, the phase delay can be closely approximated by \cite{Narbudowicz9674846}
\begin{IEEEeqnarray}{lCl}
    \theta_n(\phi) &=& \frac{\pi D}{\lambda} \operatorname{Re}\left(  \exp\left(  2 \pi i \left(\frac{n}{(N+1)}+ \phi \right)\right) \right) \nonumber \\
    &=& \frac{\pi D}{\lambda} \cos \left( \frac{n}{(N+1)}+ \phi \right),
\end{IEEEeqnarray}
where $\lambda$ is the wavelength. 
In addition, we assume that the legitimate receiver is in the direction of $\phi=\pi$, i.e. $\phi_{\text{Bob}}=\pi$, \ac{QPSK} modulation, $\pi D/\lambda = 1$, and $n$ are equiprobable for all the following cases.
Hereinafter, the \ac{BER} will be estimated considering that the receiver has properly completed the handshake and applies a trivial demodulation method. 
That is, it simply demodulates the signal assuming that the modulation is the traditional \ac{QPSK}.
\subsection{Case I: $N=2$}
\cref{fig:Const1} depicts two constellations on the receiver. 
One is in the direction $\phi = 60\degree$, see the blue circles, and the other is in the direction $\phi = 120\degree$, see the red $\times$'s. 
Constellations have been presented with different energy configurations for the sake of better visibility.
The ordered pair next to the symbols consists of the transmitted secret symbol $S$ and the transmitting antenna $n$ as $(Sn)$ with $S \in \{0,1,2,3\}$ and $n \in \{1,2\}$. 
For example, the ordered pair denoted by $12$ means the secret symbol $S$ is $1$ and the transmitting antenna $n$ is $2$.
These two constellations were chosen because they represent the extreme cases. 
In both constellations, the modulation process produces symbols $m=c_{\phi}(s,n)$ identical for different values of $n$.
However, the symbols $m$ that are superimposed for $\phi=60\degree$ represent distinct secret symbols $S$,
while overlapping symbols $m$ for $\phi=120\degree$ represent the same secret symbol $S$.
For instance, we have $c_{60\degree}(3,2)=c_{60\degree}(1,1)$ and $c_{120\degree}(1,2)=c_{120\degree}(1,1)$; therefore, it is impossible to decide whether the transmitted symbol is $s=3$ or $s=1$ for direction $\phi=60\degree$ if it receives $c_{60\degree}(3,2)$, whereas an eavesdropper in direction $\phi=120\degree$ can properly recover $S$.
This distinction can be explained as follows.
An eavesdropper in the direction $\phi=120\degree$ obtains two bits of information whenever it correctly estimates $m$, while another eavesdropper in the direction $\phi=60\degree$ only gets one bit of information on average --- the uncertainty about $S$ is reduced by $1$ bit (to two possibilities) for direction $\phi=60\degree$ given reception of $m$.
\begin{figure}[t!]
	\centering{\includegraphics[width=0.6\linewidth]{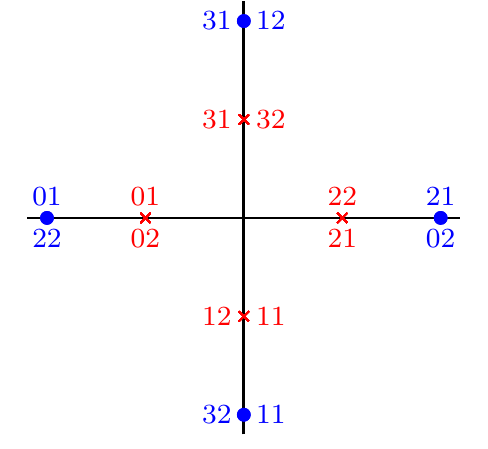}}
	\caption{Constellations on the receiver in the directions of $\phi = 60\degree$ (blue circles) and $\phi = 120\degree$ (red $\times$). The ordered pair next to the symbols is constituted by the transmitted secret symbol and by the antenna used as $(Sn)$ with $S \in \{0,1,2,3\}$ and $n \in \{1,2\}$.}
	\label{fig:Const1}
\end{figure}

Let us now introduce the \ac{AWGN} such that the \ac{SNR} equals 10 dB. 
This scenario is depicted in \cref{fig:ber1}.
Following the constellations as previously discussed, we focus on the following two cases: $\phi=60\degree$ and $\phi=120\degree$.
As aforementioned, there is an information leakage of approximately 2 bits and 1 bit in the direction $\phi=120\degree$ and $\phi=60\degree$, respectively.
However, information leakage cannot be assessed directly through the \ac{BER}.
In other words, the \ac{BER} can somewhat conceal the potential vulnerability of the system.
Note that although the information leakage for $\phi=120\degree$ is higher, the \ac{BER} for the case $\phi=60\degree$ is lower. 
This counter-intuitive effect can be explained as follows: the information leakage is calculated assuming that the eavesdropper applies \cref{th: optmal d(r)}, while the \ac{BER} calculation cannot incorporate the peculiar characteristics of the system.

\cref{fig:ber1} also reveals the sensitivity (in the context of vulnerability) of the system with regard to several directions. 
As a result, we notice a high number of directions in which the vulnerability is maximum.
At best, the system provides information leakage of around 1 bit.
In general terms, we can say that the system with three antennas is vulnerable to attacks.
Next, we consider a large value of $N$ for the sake of completeness.
\begin{figure}[t!]
\centering{\includegraphics[width=0.5\linewidth]{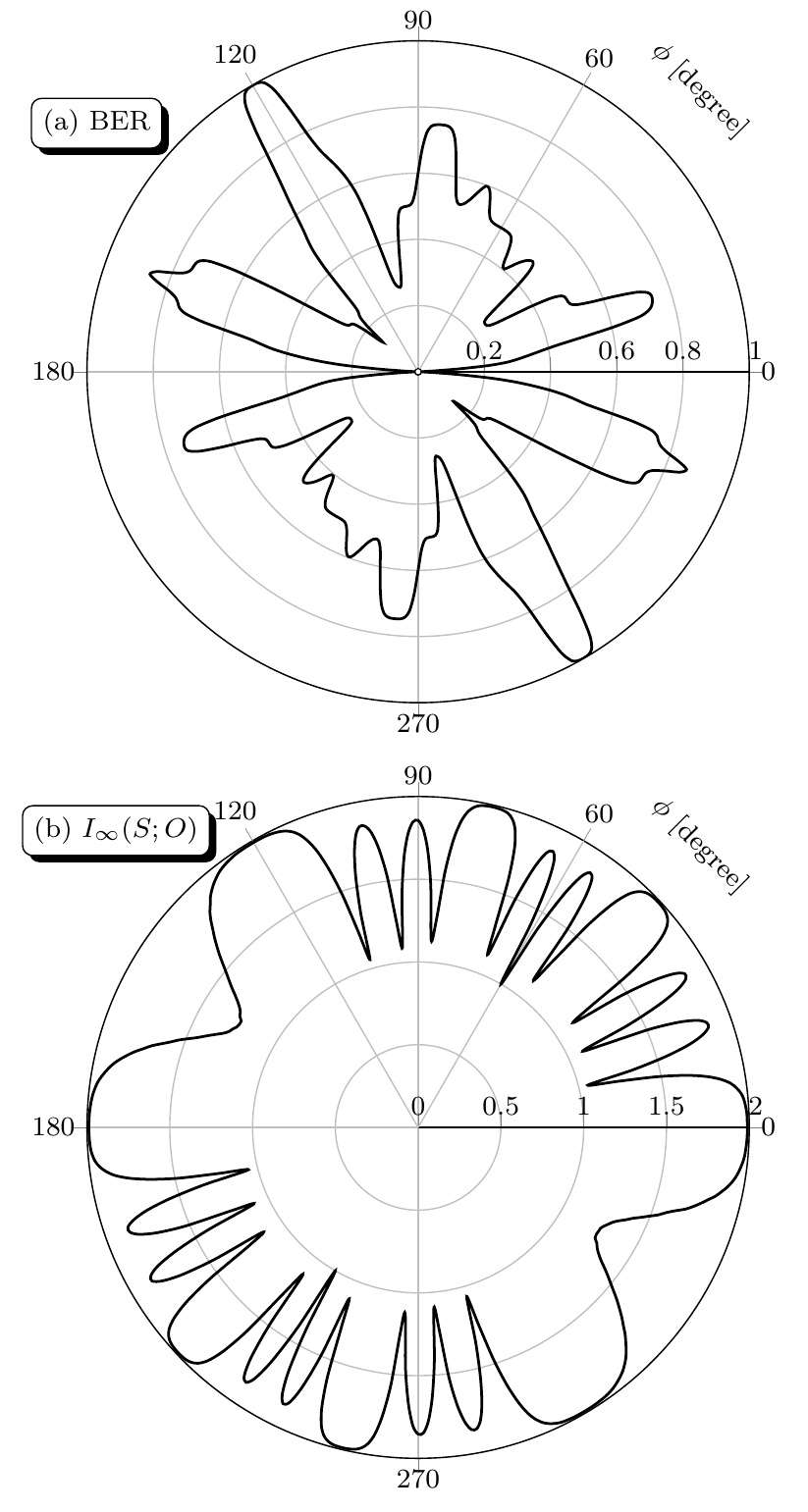}}
	\caption{\ac{BER} and information leakage $I_{\infty}(S;O)$ versus the direction $\phi$ for case $N=2$.}
	\label{fig:ber1}
\end{figure}
\subsection{Case II: $N=4$}
\cref{fig:Const2} shows the constellation on the receiver in the direction $\phi = 10\degree$. 
For this constellation, we can say that there are 4 sets of 4 symbols each, in which the symbols are very close together.
For example, the ordered pairs 11, 33, 22, and 04 constitute one set.
For any of the four sets, it can become almost infeasible in the presence of noise, to distinguish between the 4 symbols belonging to a set.
Furthermore, each of the four symbols of a set comes from a different $S$.
This characteristic of the constellation transfers little or almost no information to the receiver about the secret input $S$ when operating under noise.
Therefore, it is to be expected that it will lead to greater security for communication.
\begin{figure}[t!]
	\centering{\includegraphics[width=0.6\linewidth]{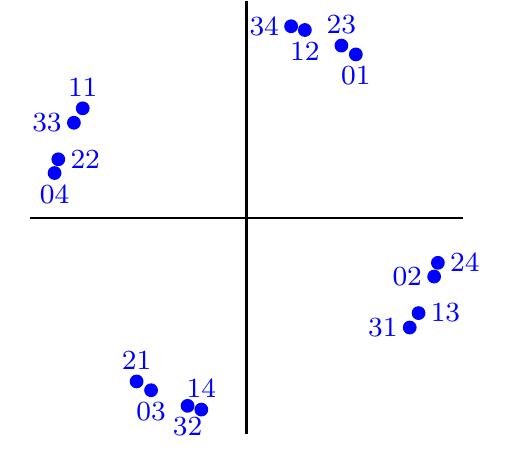}}
	\caption{Constellation on the receiver in the direction $\phi = 10\degree$. The ordered pair next to the symbols is constituted by the transmitted secret symbol and by the antenna used as $(Sn)$ with $S \in \{0,1,2,3\}$ and $n \in \{1,2,3,4\}$.}
	\label{fig:Const2}
\end{figure}

In order to verify the efficiency of the constellation shown in \cref{fig:Const2}, we assume that the system is subject to an \ac{SNR} of $10$ dB.
\cref{fig:ber2} depicts the \ac{BER} and information leakage for this case.
Note that the information leakage is low for $\phi=10\degree$.
Again, the \ac{BER} does not satisfactorily measure the system's behavior in terms of its vulnerability.
As can be seen in the figure, there is no symmetry between the \ac{BER} and information leakage. 
There is a significant reduction in the number of spikes in information leakage as $N$ increases from 2 to 4.
Furthermore, we obtain values below 1 for the information leakage with four antennas.
On the other hand, in general the system is still somewhat vulnerable, as the information leakage is mostly above 1 bit, with a significant portion above 1.5 bits.
The following scenario discusses the case with a higher number of antennas.
\begin{figure}[t!]
	\centering{\includegraphics[width=0.5\linewidth]{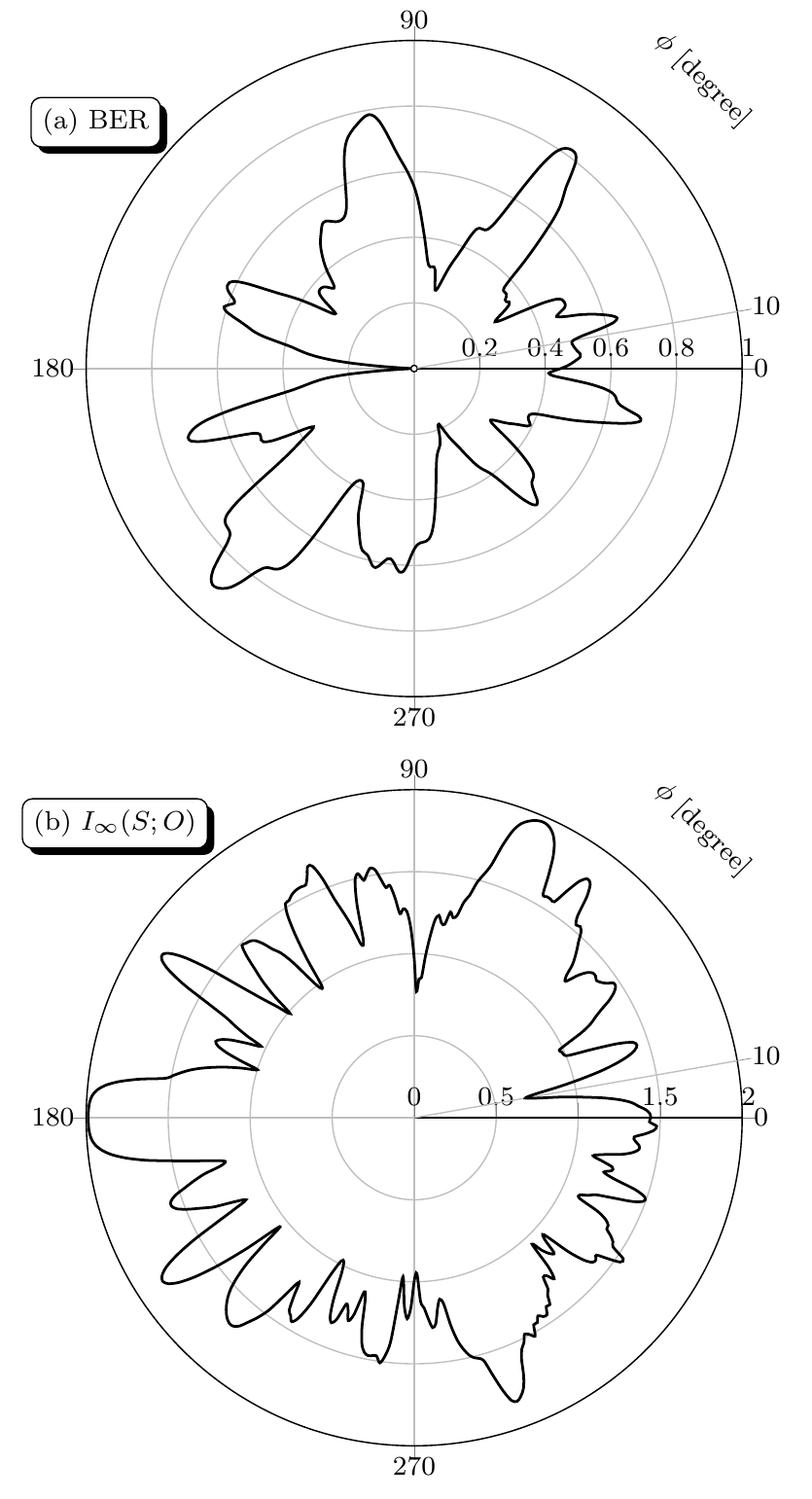}}
	\caption{\ac{BER} and information leakage $I_{\infty}(S;O)$ versus the direction $\phi$ for case $N=4$.}
	\label{fig:ber2}
\end{figure}

\subsection{Case III: $N=14$, $N=24$, $N=34$ and $N=299$}
\cref{fig:inf_leak} depicts the \ac{BER} and information leakage versus the direction $\phi$ for different values of $N$ at an \ac{SNR} of $10$ dB. 
Furthermore, it can be seen that the average values of the information leakage are $0.8641$, $0.7733$, $0.7175$, and $0.6720$ for $N=14$, $N=24$, $N=34$, and $N=299$, respectively. 
It is evident that an increase in the value of $N$ yields a decrease in the information leakage in certain directions of $\phi$.
On the other hand, no significant change in \ac{BER} can be observed with respect to a variation in the number of antennas.
Note that although there is a substantially large jump in the number of antennas from $N=35$ to $N=299$, we do not obtain a significant improvement in the vulnerability of the system.
\begin{figure}[t!]
	\centering{\includegraphics[width=0.5\linewidth]{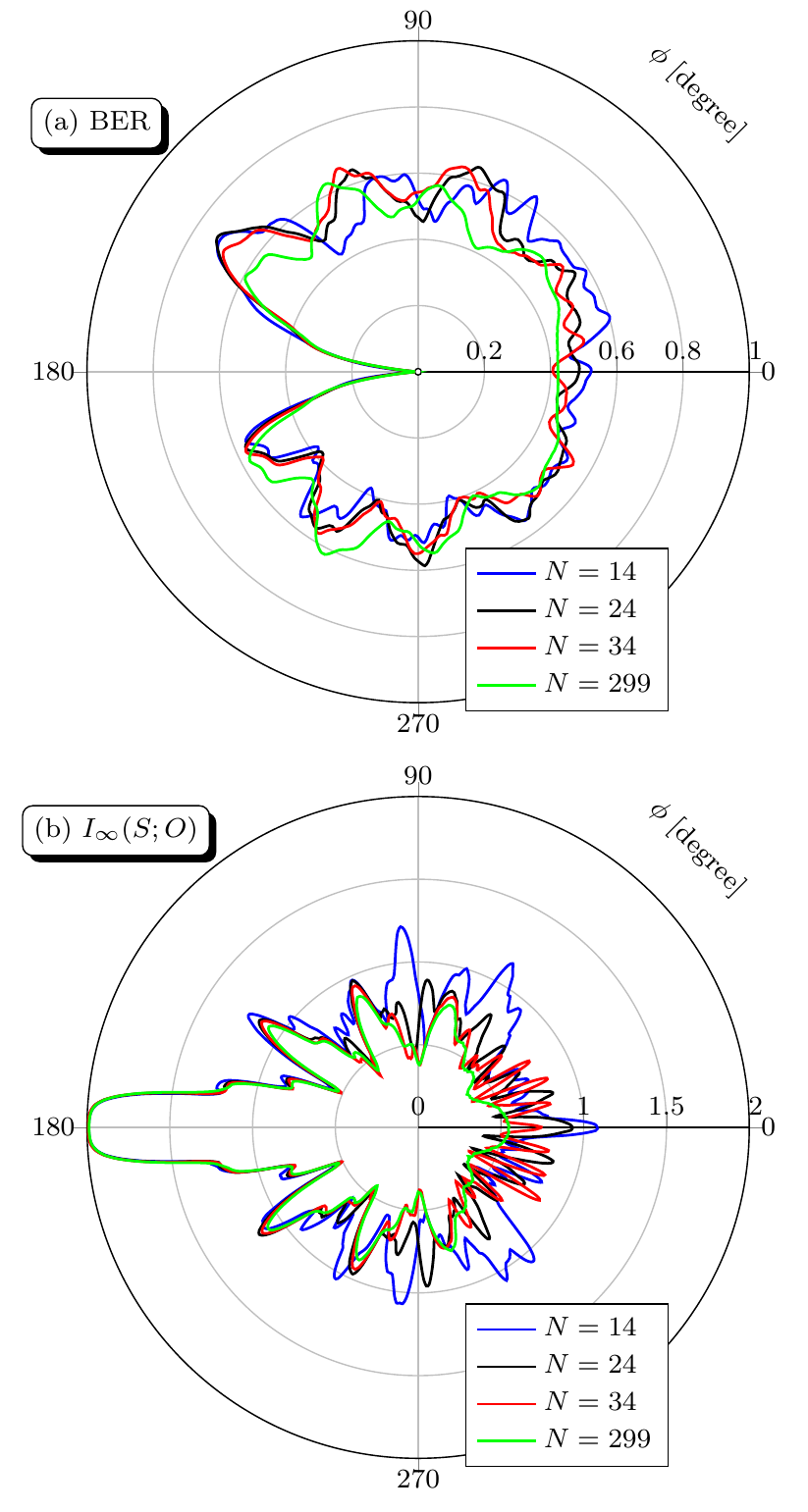}}
	\caption{\ac{BER} and information leakage $I_{\infty}(S;O)$ versus the direction $\phi$ for $N=\{14,24,34,299\}$. The mean values of $I_{\infty}(S;O)$ are $0.8641$, $0.7733$, $0.7175$, and $0.6720$ for $N=14$, $N=24$, $N=34$, and $N=299$, respectively.}
	\label{fig:inf_leak}
\end{figure}

\subsection{Case IV: AWGN dependence}
We now turn our attention to the interdependence of information leakage with the \ac{AWGN} as the last case study.
Unlike the case studies previously presented, the number of antennas is fixed at $N=6$, and the \ac{SNR} varies as follows $\text{SNR}=\{9,15,\infty\}$ dB.
\cref{fig:inf_leak_x_SNR} depicts the \ac{BER} and information leakage for the various \ac{SNR} values.
The mean values of $I_{\infty}(S;O)$ are $1.0712$, $1.4052$, and $2$ for $\text{SNR}=9$ dB, $\text{SNR}=15$ dB, and $\text{SNR}\to \infty$, respectively.
As previously stated, we have a maximum information leakage in any direction in the absence of noise;
in mathematical terms, $\lim_{\sigma \to 0} I_{\infty}(S;O) = H_{\infty}(S)$ and thus $\lim_{\sigma \to 0} I_{\infty}(S;O) = 2$ bits for equiprobable symbols with \ac{QPSK}, corroborating our theoretical deductions.
Again, no significant change in the performance can be noticed through the \ac{BER} for the evaluated \ac{SNR} values; while the leakage information has a clear improvement with the decrease of the \ac{SNR}.
Note that the leakage information reaches the maximum and \ac{BER} tends to 1 in the direction $\phi=260\degree$ for $\text{SNR}=15$ dB.
This case study highlights the interdependence between noise and the vulnerability of \ac{DDM}.
\begin{figure}[b!]
	\centering{\includegraphics[width=0.5\linewidth]{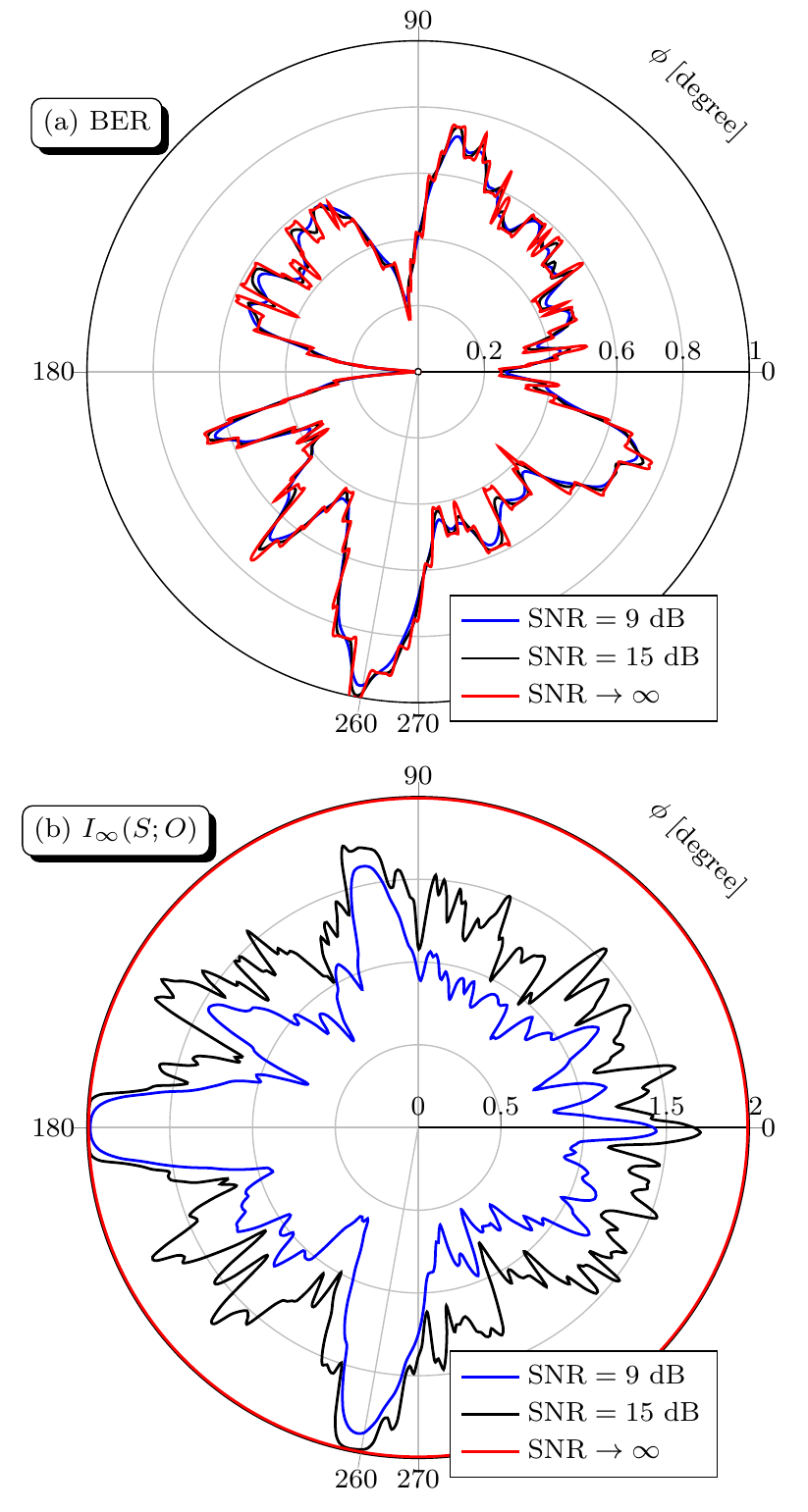}}
	\caption{\ac{BER} and information leakage $I_{\infty}(S;O)$ versus the direction $\phi$ for $\text{SNR}=\{9,15, \infty \}$ dB with $N=6$. The mean values of $I_{\infty}(S;O)$ are 1.0712, 1.4052, and 2 for $\text{SNR}=9$ dB, $\text{SNR}=15$ dB, and $\text{SNR}\to \infty$, respectively.}
	\label{fig:inf_leak_x_SNR}
\end{figure}

Given the four previous cases, we propose the following statement.
\begin{conjecture}\label{conj:1}
    The closer the values of $c_{\mathcal{A}}(s,n) \, \forall s \in \mathcal{S}$ for the same $n$ are, the less vulnerable the wireless system becomes.
\end{conjecture}
\Cref{conj:1} can be explained, among other ways, as follows. 
The system's vulnerability depends on the format of the constellation generated by the \ac{DDM}. 
In addition, constellations that coalesce symbols from the same secret input $S$ are more vulnerable. 
Therefore, two design guidelines can be established: (i) given an antenna array, security can be increased in a specific direction by selecting for transmission only those antennas that meet \cref{conj:1}; (ii) the antenna array design should be based on \cref{conj:1}.

\section{Conclusion}
\label{sec: Conclusion}
This paper has assessed the vulnerability of \ac{DDM} by evaluating of the information leakage measure, and has also addressed the fundamental limits of the information leakage.
Despite its usefulness as a performance metric for the evaluation of a wireless communication system, the proposed study has revealed the inefficiency of the \ac{BER} as a security measure as it can potentially hide the system vulnerability.
Moreover, it is shown that an eavesdropper is able to remove the maximum amount of information from the transmitted signal in some directions by meeting \cref{th: optmal d(r)} or \cref{th: Inf lea}. 
With this (optimized) eavesdropping model, several proposed guidelines on the design of the antenna array that minimizes the information leakage, are presented.
Sufficient conditions for a null information
leakage in the direction $\phi \neq \phi_{\text{Bob}}$ are presented in \cref{cl: leak to 0}.
Although the proposed work provides some useful insights, it assumes an ideal eavesdropping model wherein the eavesdropper has perfect knowledge of the channel statistics, which is impractical in realistic models. 
One could investigate the impact of imperfect channel estimation on the eavesdropper's performance as a possible future extension of the work presented here.

%

	\bibliographystyle{IEEEtran}
	\bibliography{IEEEabrv,References}

\end{document}